\documentclass[11pt,fleqn]{article}
\usepackage[round]{natbib}
\usepackage{amsmath,amssymb,amsthm}
\usepackage{xcolor}
\usepackage[pdfstartview=FitH,bookmarksopen=false,
colorlinks=true,linkcolor=blue,citecolor=blue]{hyperref}
\usepackage[left=2cm,right=2cm,top=2cm,bottom=2cm]{geometry}

\newtheorem{theorem}{Theorem}

\newtheorem{corollary}{Corollary}
\newtheorem{lemma}{Lemma}
\newcounter{mycounter}
\newtheorem{counterexample}{Counterexample}

\begin{document}
	
	\title{\bf A note on the Diversity Owen value\thanks{This work is supported by the National Natural Science Foundation of China (No. 72371151)} }
	\author{Songtao He,\,  Erfang Shan\thanks{{\em Corresponding author}. E-mail addresses: efshan@shu.edu.cn (E. Shan), hesongtao@shu.edu.cn (S. He), 23720622@shu.edu.cn (Y. Sun)},\, Yuxin Sun\\
		{\small School of Management, Shanghai University, Shanghai 200444, P.R. China}}
	\date{}
	\maketitle \baselineskip 17pt
	
	\begin{abstract}
		
		B\'eal et al. (Int J Game Theory 54, 2025) introduce the Diversity Owen value for TU-games with diversity constraints, and provide  axiomatic characterizations using the axioms of fairness and balanced contributions. However, there exist logical flaws in the proofs of the uniqueness of these characterizations. In this note we provide the corrected proofs of the characterizations by introducing the null player for diverse games axiom. Also, we establish an alternative characterization of the Diversity Owen value by modifying the axioms of the above characterizations.
\bigskip
		
\noindent {\bf Keywords}: Diversity Owen value; Axiomatization
		
\medskip
		\noindent {\bf AMS (2000) subject classification:} 91A12
		
		\noindent {\bf JEL classification:} C71
	\end{abstract}
	
\section{Introduction, definitions and notation}\label{section2}
The Diversity Owen value for TU-games with diversity constraints is introduced and axiomatized by \cite{2025_beal2}.
In this note we correct the proofs of the uniqueness in the axiomatizations of the  value in \cite{2025_beal2} and provide an alternative characterization.
Throughout this note, we follow the definitions and notation in \cite{2025_beal2}. Here we only give the necessary terminology and notation. For more detailed explanations, see \cite{2025_beal2}.

	A cooperative game with transferable utility (TU-game) is a pair $(N, v)$ where $N$ is	a finite set of players and $v$ is a function assigning a worth $v(S)\in  \mathbb{R}$ to each coalition $S\subseteq N$ and such that $v(\emptyset) = 0$. A player $i$ is {\em null} in $v$ if $v(S \cup i)=v(S)$ for all $S \subseteq N \setminus \{i\}$. Two players $i$ and $j$ are {\em symmetric} in $v$ if $v(S \cup i) - v(S) = v(S \cup j) - v(S)$ for all $S \subseteq N \setminus \{i, j\}$.
	For nonempty $S \subseteq N$, $(S, v_{|S})$ is the subgame of $v$ to $S$ defined by $v_{|S}(T) = v(T)$ for all $T \subseteq S$. For every coalition $S\subseteq N$, the {\em unanimity game} $(N, u_S)$ is defined by $u_S(T) = 1$ if $S \subseteq T$ and $u_S(T) = 0$ otherwise.
	Any TU-game can be uniquely represented by unanimity games,
	\[v = \sum_{T \subseteq N, T \neq \emptyset} \Delta_T(v) u_T,\]
	where $\Delta_T(v) = \sum_{R\subseteq T} (-1)^{|T|-|R|} v(R)$ is called the {\em Harsanyi dividend} of the coalition $T$ \citep{1959_harsanyi_bar}.
	
	A cooperative game with a coalition structure is a triple $(N, v, \mathcal{B})$ where $(N, v)$ is a
	TU-game and $\mathcal{B} =\{B_1, B_2, \cdots, B_m\}$ is a coalition structure, i.e., $\cup_{k\in M}B_k=N$, where $M=\{1,2,\cdots\,m\}$. The elements of $\mathcal{B}$ are called components. For any $S\subseteq N$, we denote by $\mathcal{B}_{|S} =\{B_1 \cap S, \cdots, B_m \cap S\}$ the coalition
	structure on $S$ induced by $\mathcal{B}$. For any $S\subseteq N$, the subgame with a coalition structure
	on $S$ induced by $\mathcal{B}$ is $(S, v_{|S}, \mathcal{B}_{|S})$.

Recently, \cite{2025_beal1} extend TU-games with coalition structures to TU-games with diversity constraints.
They introduce a new allocation rule --
	the Diversity Owen value for TU-games with diversity constraints, which offers significant relevance in various economic and political contexts.
	
	A \textit{TU-game with diversity constraints} is a four-tuple $(N, v, \mathcal{B}, d)$,
	where $(N, v, \mathcal{B})$ is a TU-game with a coalition structure composed of communities,
	and $d = (d_1, \dots, d_m)$ specifies, for each $k \in \{1, \dots, m\}$, the minimal number $d_k \in \{1, \dots, |B_k|\}$
	of members of community $B_k$ that a coalition should contain to be considered as \emph{diverse}.
	
	Hence, a given coalition $S \subseteq N$ is \emph{diverse} if $|S \cap B_k| \geq d_k$ for each $k \in \{1, \dots, m\}$,
	and we denote by $D(N, \mathcal{B}, d)$ the set of all diverse coalitions.
	Denote by $GD$ the set of all TU-games with diversity constraints and a finite player set.
	
	A game $(N, v, \mathcal{B}, d)$ is called \emph{diverse} if $v(S) \neq 0$ implies that $S \in D(N, \mathcal{B}, d)$.
	The \emph{diversity-restricted game} associated to $(N, v, \mathcal{B}, d)$ is the game $(N, v^d, \mathcal{B}, d)$ such that
	\[
	v^d(S) =
	\begin{cases}
		v(S) & \text{if } S \in D(N, \mathcal{B}, d), \\
		0 & \text{otherwise}.
	\end{cases}
	\]
	By construction, $(N, v^d, \mathcal{B}, d)$ is a diverse game.
	
	A player $i \in B_k \in \mathcal{B}$ is  \emph{out} in a game $(N, v, \mathcal{B}, d)$ if its presence is not necessary to achieve the diversity requirement, i.e., $|B_k| - d_k \geq 1$.
	A game $(N, v, \mathcal{B}, d)$ is called \emph{$i$-out diverse}, $i \in B_k$, if it is diverse and $|B_k| - d_k \geq 1$. For any $S\subseteq N$, the subgame with diversity constraints of $(N, v, \mathcal{B}, d)$
	induced by $S$ is the game with diversity constraints $(S, v_{|S}, \mathcal{B}_{|S}, d)$.
	
	
	Formally, for any $(N, v, \mathcal{B}, d)\in GD$, the \textbf{Diversity Owen value} is defined by
	\begin{center}
		$
		DOw(N, v, \mathcal{B}, d) = Ow(N, v^d, \mathcal{B}).
		$
	\end{center}
	
	We list the axioms in \cite{2025_beal2} used to characterize the Diversity Owen value.
	
	\textbf{Efficiency (E).} For each $(N, v, \mathcal{B}, d) \in GD$,
	$\sum_{i \in N} f_i(N, v, \mathcal{B}, d) = v(N)$.
	
	\textbf{Null Player Out for Preserving-Diversity Games (NPOPD).} If $(N, v, \mathcal{B}, d)$ is $i$-out diverse and player $i$ is null in $(N, v, \mathcal{B}, d)$, then $f_j(N, v, \mathcal{B}, d) = f_j(N \setminus \{i\}, v |_{N \setminus \{i\}}, \mathcal{B} |_{N \setminus \{i\}}, d)$ for each $j \in N \setminus \{i\}$.
	
	\textbf{Fairness within Component (FwC).} For all $(N, v, \mathcal{B}, d),\ (N, w, \mathcal{B}, d) \in GD$ and $i, j \in B_p \in \mathcal{B}$ symmetric in $(N, v)$, then
	\[ f_i(N, v + w, \mathcal{B}, d) - f_i(N, w, \mathcal{B}, d) = f_j(N, v + w, \mathcal{B}, d) - f_j(N, w, \mathcal{B}, d).\]
	
	\textbf{Fairness through Diversity (FD).} For all $(N, v, \mathcal{B}, d),\ (N, w, \mathcal{B}, d) \in GD$ and $B_p, B_q \in \mathcal{B}$,
	\[ \sum_{i \in B_p} f_i(N, v + w, \mathcal{B}, d) - \sum_{i \in B_p} f_i(N, w, \mathcal{B}, d)
	= \sum_{i \in B_q} f_i(N, v + w, \mathcal{B}, d) - \sum_{i \in B_q} f_i(N, w, \mathcal{B}, d). \]
	
	\textbf{Independence from Non-Diverse Coalitions (INDC).}
	For each $(N, v, \mathcal{B}, d),\ (N, w, \mathcal{B}, d) \in GD$ such that  $v(S) = w(S)$ for all $S \in D(N, \mathcal{B}, d)$,
	$f(N, v, \mathcal{B}, d) = f(N, w, \mathcal{B}, d).$
	
	\textbf{Equality through diversity (ED).}  ~For any  $(N, v, \mathcal{B}, d) \in GD$ and each $k,q\in M$,
	$\sum_{i\in B_k}f_i(N, v, \mathcal{B}, d)=\sum_{i\in B_q}f_i(N, v, \mathcal{B}, d)$.

	\textbf{Intra-coalitional balanced contributions with out players for preserving diversity (IBCOPPD).}  For any diverse game $(N, v, \mathcal{B}, d) \in GD$ and any out players $i,j\in B_p \in \mathcal{B}$ with $i\neq j$,
	$
	f_i(N, v, \mathcal{B}, d) - f_i(N\setminus \{j\}, v_{|N\setminus \{j\}}, \mathcal{B}_{|N\setminus \{j\}}, d)=f_j(N, v, \mathcal{B}, d) - f_j(N\setminus \{i\}, v_{|N\setminus \{i\}}, \mathcal{B}_{|N\setminus \{i\}}, d).
	$ Moreover, if player $i$ is null in $(N\setminus \{j\}, v_{|N\setminus \{j\}}, \mathcal{B}_{|N\setminus \{j\}}, d)$ then $f_i(N\setminus \{j\}, v_{|N\setminus \{j\}}, \mathcal{B}_{|N\setminus \{j\}}, d)=0$ (respectively, if $j$ is null in $(N\setminus \{i\}, v_{|N\setminus \{i\}}, \mathcal{B}_{|N\setminus \{i\}}, d)$ then $f_j(N\setminus \{i\}, v_{|N\setminus \{i\}}, \mathcal{B}_{|N\setminus \{i\}}, d)=0$).

\cite{2025_beal2} obtain the following main results on axiomatizations of the Diversity Owen value.
\begin{theorem}\label{thm1}
		(\citeauthor{2025_beal2} \citeyear{2025_beal2}, Proposition 2) The Diversity Owen value is the unique value on GD that satisfies Efficiency (\textbf{E}), Fairness within Component (\textbf{FwC}), Fairness through Diversity (\textbf{FD}), Independence from Non-Diverse Coalitions (\textbf{INDC}), and Null Player Out for Preserving-Diversity Games (\textbf{NPOPD}).
	\end{theorem}
\begin{theorem}\label{thm2}
		(\citeauthor{2025_beal2} \citeyear{2025_beal2}, Proposition 4) The Diversity Owen value is the unique value on $GD$ that satisfies Efficiency (\textbf{E}), Intra-coalitional balanced contributions with out players for preserving diversity (\textbf{IBCOPPD}), Equality through diversity(\textbf{ED}), and Independence from Non-Diverse Coalitions (\textbf{INDC}).
	\end{theorem}
	
However, we notice that there are some flaws in the proofs of the uniqueness of these two results in \cite{2025_beal2}.
In this note, we will identify the errors in the uniqueness proofs of both results and present the corrected proofs.
 Additionally, we introduce the null player for diverse games axiom, and  provide an alternative characterization of the Diversity Owen value by replacing the null player out for preserving-diversity games axiom used in Theorem \ref{thm1} with this axiom.
	\section{On Theorem \ref{thm1}}\label{section3}
Before  giving the corrected  proof of the uniqueness of Theorem \ref{thm1} in \cite{2025_beal2}, we first point out the errors  in the original proof of the uniqueness in Theorem \ref{thm1}.
	
	\begin{quote}
		\textbf{Subcase 1a.} ``Assume that $B_h = \mathcal{B}(i) = \{i, i'\}$. So, $d_h$ is necessarily equal to $1$.
		There exists a diverse coalition $T_1 \in \mathcal{T}_{v^d}$ such that $i' \notin T_1$.
		Moreover, $i \in T_1$ and $i' \in T_0$ since $T_1$ and $T_0$ are diverse.
		The players $i$ and $i'$ are symmetric in
		$
		(N,\, \Delta_{v^d}(T_0)u_{T_0} + \Delta_{v^d}(T_0)u_{T_1}).
		$ By applying \textbf{FwC} $\cdots$''
		
		\textbf{Subcase 1b.} ``$\cdots$Recall that $T_0 \neq T_2$. So there are at least \( |\mathcal{B}(i)| - 1 \) linearly independent equations in Eqs. (3), (4), and (5).
		Combining these equations with Eq.(1) leads to
		$f_i(N, v^d, \mathcal{B}, d) = Ow_i(N, v^d, \mathcal{B}).$''
		
	\end{quote}
	
	In Subcase 1a, it is possible that the players \( i \) and \( i' \) are not symmetric in $
	(N,\, \Delta_{v^d}(T_0)u_{T_0} + \Delta_{v^d}(T_0)u_{T_1})
	$.  This situation will cause the application of the {\bf FwC} axiom to fail. In Subcase 1b, if $T_0 \cap \mathcal{B}(i)=\mathcal{B}(i)\setminus T_2$,
In this case, it would be impossible to construct
\( |\mathcal{B}(i)| - 1 \)
 linearly independent systems of equations.

The following counterexamples illustrate those mentioned above.
	\begin{counterexample}\label{counter1}
		Consider a game \((N, v, \mathcal{B}, d) \in GD\), where \(N = \{i, j, k, l\}\), \(\mathcal{B} = \{B_h , B_q\} \), \(B_h = \{i, j\} \), \( B_q = \{k, l\}\), and \(v = \Delta_v(\{i,k\}) u_{\{i,k\}} + \Delta_v(\{j,l\}) u_{\{j,l\}}\), with diversity requirement \(d = (1,1)\). It is easy to verify that this game satisfies the conditions of Subcase 1a of Case 1 in \cite{2025_beal2}, i.e., \(|B_h| = 2\) and \(\mathcal{T}(v^d) = \emptyset\).
		
	\end{counterexample}
	Define 	$w = \Delta_v(\{i,k\}) u_{\{i,k\}} + \Delta_v(\{i,k\}) u_{\{j,l\}}$. Obviously, $i$ and $j$ are not symmetric in $w$.
	
In fact, the method in \cite{2002_Hu} cited by the authors are also wrong.  We construct the game $w'$ to address this. Let
    $w' = \Delta_{v}(\{j,l\}) u_{\{j,l\}} + \Delta_{v}(\{j,l\}) u_{\{i,l\}}.
	$
	It is easy to see that \( i \) and \( j \) are symmetric in $w'$. However, since \( |\mathcal{I}(v - w')|=|\mathcal{I}(v)| \), it is not possible to apply induction to obtain
	$
	f_i(N, v, \mathcal{B},d) - f_i(N, v, \mathcal{B},d) = DOw_i(N, v, \mathcal{B},d) - DOw_i(N, v, \mathcal{B},d).
	$

	\begin{counterexample}
		Consider a game \((N, v, \mathcal{B}, d) \in GD\), where \(N = \{i, j, k, l\}\), \(\mathcal{B} = \{B_h, B_q\}\), \( B_h = \{i, j, k\}\), \( B_q = \{l\} \), and \(v = \Delta_v(\{i\}) u_{\{i\}} + \Delta_v(\{j,k,l\}) u_{\{j,k,l\}}\), with diversity requirement \(d = (1,1)\). It is easy to verify that this game satisfies the conditions of Subcase 1b of Case 1 in \cite{2025_beal2}, i.e., \(|B_h| > 2\) and \(\mathcal{T}(v^d) = \emptyset\).
	\end{counterexample}
	We can only conclude that $j$ and $k$ are symmetric in both $\Delta_v(\{j,k,t\}) u_{\{j,k,t\}}$ and $\Delta_v(\{i\}) u_{\{i\}}$, which is insufficient to obtain at least $|B_h|-1$ linearly independent equations.
	In fact, in the process of \citeauthor{2002_Hu}'s characterization of the Owen value using fairness axioms, a similar issue arises at the player level. It is also impossible to derive $|B_h|-1$ linearly independent equations.

	\subsection{The corrected proof of the uniqueness of Theorem \ref{thm1}}
	In this subsection we give the corrected  proof of the uniqueness of Theorem \ref{thm1} in \cite{2025_beal2}. For this purpose, we first introduce the following axiom and lemma.
	
	\textbf{Null Player for Diverse Games (ND).} If $(N, v, \mathcal{B}, d)\in GD$ is diverse and player $i\in N$ is a null player in $(N, v, \mathcal{B}, d)$, then $f_i(N, v, \mathcal{B}, d) =0$.
	
	\begin{lemma}\label{lemma1}
		{\bf E} and {\bf NPOPD} imply {\bf ND}
	\end{lemma}
	
	\begin{proof}
		 Let $f$ be a value on $GD$ that satisfies {\bf E} and {\bf NPOPD}. For any diverse game $(N, v, \mathcal{B}, d) \in GD$, player $i \in B_k\in \mathcal{B}$ is a null player in $v$, we show that $f_i(N,v,\mathcal{B},d)=0$.
		
Take an outside player $l\notin N$, and we define the extended game $(N \cup \{l\},(v)_{+l},\mathcal{B}_{+l},d)$  by
\begin{align*}
&(v)_{+l}(S)=v(S\setminus \{l\}) \ \ \ \mbox{for all $S\subseteq N\cup \{l\}$},\\
 &\mathcal{B}_{+l}=\{B_1, B_2, \cdots, B_{k'}, \cdots, B_m\}, \ \ \mbox{where $B_{k'}=B_k\cup \{l\}$.}
 \end{align*}
		Note that $l$ is null in $(v)_{+l}$. 	For any $S\subseteq N\cup \{l\}$, if $(v)_{+l}(S)\neq 0$, then $v(S\setminus \{l\})\neq 0$.  Thus $S\setminus \{l\} \in D(N,\mathcal{B},d)$, i.e., $|S\setminus \{l\} \cap B_{p} |\ge d_p$ for all $B_p \in \mathcal{B}$. Since  $|S\cap B_{p} |\geq| S\setminus \{l\} \cap B_{p} |\ge d_p$ for all $B_{p}\in \mathcal{B}_{+l}$,  the game $(N \cup \{l\},(v)_{+l},\mathcal{B}_{+l},d)$ is also a diverse game.
		Combining $|B_{k'}|-d_k\ge |B_{k'}|-|B_{k}|\ge 1$, it follows that  $(N \cup \{l\},(v)_{+l},\mathcal{B}_{+l},d)$ is $i$-out diverse and $l$-out diverse.
		
		Thus, by {\bf NPOPD} and {\bf E}, we have
		\begin{eqnarray*}
			v(N)&=&(v)_{+l}(N\cup \{l\})\\
&\stackrel{{\bf E}}{=}&f_i(N \cup \{l\},(v)_{+l},\mathcal{B}_{+l},d)+\sum_{j \in N \backslash \{i\}}f_j(N \cup \{l\},(v)_{+l},\mathcal{B}_{+l},d)\nonumber \\ &\stackrel{{\bf NPOPD}}{=}&f_i(N,v,\mathcal{B},d)+\sum_{j \in N \backslash \{i\}}f_j\big((N \cup \{l\})\setminus \{i\},{(v)_{+l}}_{|(N\cup \{l\} )\setminus \{i\}},{\mathcal{B}_{+l}}_{|(N\cup \{l\}) \setminus \{i\}},d \big) \nonumber \\
			&\stackrel{{\bf E}}{=}&f_i(N,v,\mathcal{B},d)+(v)_{+l}\big((N\cup \{l\}) \setminus \{i\}\big) \nonumber \\
			&=&f_i(N,v,\mathcal{B},d)+v( N\setminus \{i\}).
		\end{eqnarray*}
		Hence, $f_i(N,v,\mathcal{B},d)=0$.
	\end{proof}
	
	We now provide a corrected  proof of the uniqueness in Theorem \ref{thm1}.
	
	\begin{proof}[Proof of Theorem~\ref{thm1}]
	In order to prove the uniqueness part,	let \( f \) be a value on \( GD \) that satisfies the five axioms.  We show that \( f(N, v, \mathcal{B}, d) = DOw(N, v, \mathcal{B}, d) \) for any $(N, v, \mathcal{B}, d) \in GD$. By Lemma \ref{lemma1}, $f$ satisfies {\bf ND}. By {\bf INDC}, we have \( f(N, v, \mathcal{B}, d) = f(N, v^d, \mathcal{B}, d)\) and \( DOw(N, v, \mathcal{B}, d) = DOw(N, v^d, \mathcal{B}, d)\). Thus, it suffices to show that  $f(N, v^d, \mathcal{B}, d)=$ $ DOw(N, v^d, \mathcal{B}, d)$.
		Set
		$\mathcal{I}(v^d) := \{ T \in D(N, \mathcal{B}, d) : \Delta_{v^d}(T) \neq 0 \}.$
		We show that $f(N, v^d, \mathcal{B}, d) = DOw(N, v^d, \mathcal{B}, d)$ by induction on \( |\mathcal{I}(v^d)|  \).
		
		{\em Induction basis}: If $|\mathcal{I}(v^d) |=0$, then each $i \in N$ is null in $v^d$. Since $v^d$ is a diverse game, by {\bf ND}, we have $f_i(N,v^d,\mathcal{B},d)=0=DOw_i(N,v^d,\mathcal{B},d)$ for any $i \in N$.
		
		{\em Induction hypothesis }({\bf IH}): Assume that $f(N,v^d,\mathcal{B},d)=DOw(N,v^d,\mathcal{B},d)$ for any $(N, v, \mathcal{B}, d)$ with $|\mathcal{I}(v^d)| < m$.
		
		{\em Induction step }: Let us show that the assertion holds for any $(N, v, \mathcal{B}, d)$ with $|\mathcal{I}(v^d)| = m$.
		
		We first show that $\sum_{i \in B_k}f_i(N,v^d,\mathcal{B},d)=\sum_{i \in B_k}DOw_i(N,v^d,\mathcal{B},d)$ for any $B_k \in \mathcal{B}$. By {\bf FD}, for any $B_p,B_q \in \mathcal{B}$, we have
		\[\sum_{i \in B_p}f_i(N,v^d,\mathcal{B},d)-\sum_{i \in B_q}f_i(N,v^d,\mathcal{B},d)\stackrel{{\bf FD}}{=}\sum_{i \in B_p}f_i(N,{\bf 0},\mathcal{B},d)-\sum_{i \in B_q}f_i(N,{\bf 0},\mathcal{B},d)=0.\]
		Hence,
		$\sum_{i \in B_p}f_i(N,v^d,\mathcal{B},d)=\sum_{i \in B_q}f_i(N,v^d,\mathcal{B},d)$.
		Similarly, we have
		$\sum_{i \in B_p}DOw_i(N,v^d,\mathcal{B},d)=\sum_{i \in B_q}DOw_i(N,v^d,\mathcal{B},d).$
		Combining {\bf E}, we obtain that,
		\begin{equation}\label{eq1}
			\sum_{i \in B_k}f_i(N,v^d,\mathcal{B},d)=\sum_{i \in B_k}DOw_i(N,v^d,\mathcal{B},d), \text{~for any~} B_k \in \mathcal{B}.
		\end{equation}
		
		We next show that $f_i(N, v^d, \mathcal{B}, d) = DOw_i(N, v^d, \mathcal{B}, d)$ for all $i \in B_k$. For any $B_k \in \mathcal{B}$, if $|B_k|=1$, then by \eqref{eq1}, we have $f_i(N, v^d, \mathcal{B}, d) = DOw_i(N, v^d, \mathcal{B}, d)$ for any $i \in B_k$.
		Otherwise, $|B_k| \geq 2$. We set $\mathcal{T}(v^d) = \{i \in N : i \in T \text{ for all } T \in \mathcal{I}(v^d)\}$.
		
		For any $i \in B_k \backslash \mathcal{T}(v^d)$, there exists $T \in \mathcal{I}(v^d)$ such that $i \notin T$.
		As in Lemma \ref{lemma1}, we construct the  game $(N \cup \{l\},(v)_{+l},\mathcal{B}_{+l},d)$ by adding an outside player $l\notin N$ to $B_k$.
		Note that $l$ is null in $(v^d)_{+l}$ and $(N \cup \{l\},(v^d)_{+l},\mathcal{B}_{+l},d)$ is a diverse game. Moreover,  $i$ and $l$ are symmetric in $(\Delta_{v^d}(T)u_T)_{+l}$. Let $(w^d)_{+l}= (v^d)_{+l}-(\Delta_{v^d}(T)u_T)_{+l}$. Note that $(w^d)_{+l}(S)=(w^d)(S\setminus \{l\})$ for every $S\subseteq N\cup \{l\}$. By {\bf NPOPD}, {\bf ND}, and {\bf FwC},  we have
		\begin{eqnarray}\label{eq2}
			f_i(N, v^d, \mathcal{B}, d)&\stackrel{{\bf NPOPD}}{=}&f_i(N \cup \{l\},(v^d)_{+l},\mathcal{B}_{+l},d) \nonumber \\
			&\stackrel{{\bf ND}}{=}&
			f_i(N \cup \{l\},(v^d)_{+l},\mathcal{B}_{+l},d)-f_l(N \cup \{l\},(v^d)_{+l},\mathcal{B}_{+l},d) \nonumber \\
			&\stackrel{{\bf FwC}}{=}&
			f_i(N \cup \{l\},(w^d)_{+l},\mathcal{B}_{+l},d)-f_l(N \cup \{l\},(w^d)_{+l},\mathcal{B}_{+l},d)
			\nonumber \\
			&\stackrel{{\bf IH}}{=}&
			DOw_i(N \cup \{l\},(w^d)_{+l},\mathcal{B}_{+l},d)-DOw_l(N \cup \{l\},(w^d)_{+l},\mathcal{B}_{+l},d)
			\nonumber \\
			&\stackrel{{\bf FwC}}{=}&
			DOw_i(N \cup \{l\},(v^d)_{+l},\mathcal{B}_{+l},d)-DOw_l(N \cup \{l\},(v^d)_{+l},\mathcal{B}_{+l},d)
			\nonumber \\
			&\stackrel{{\bf ND}}{=}&
			DOw_i(N \cup \{l\},(v^d)_{+l},\mathcal{B}_{+l},d)
			\nonumber \\
			&\stackrel{{\bf NPOPD}}{=}&
			DOw_i(N, v^d, \mathcal{B}, d).
		\end{eqnarray}
Hence, $f_i(N, v^d, \mathcal{B}, d)=DOw_i(N, v^d, \mathcal{B}, d)$ for all  $i \in B_k \backslash \mathcal{T}(v^d)$.
		
		For any $i \in B_k \cap \mathcal{T}(v^d)$, if $|B_k \cap \mathcal{T}(v^d)|=1$, by (\ref{eq1}), (\ref{eq2}), we immediately obtain
		\[
			f_i(N, v^d, \mathcal{B}, d)=f_i(N, v^d, \mathcal{B}, d), \text{ for all } i \in B_k.
		\]
		Otherwise, $|B_k \cap \mathcal{T}(v^d)|\ge 2$. Note that all  players  $i,j \in B_k \cap \mathcal{T}(v^d)$ are symmetric in $v^d$, and by {\bf FwC}, we have
		\begin{eqnarray*}
			f_i(N, v^d, \mathcal{B}, d) - f_j(N, v^d, \mathcal{B}, d)&\stackrel{{\bf FwC}}{=}&f_i(N, {\bf 0}, \mathcal{B}, d) - f_j(N, {\bf 0}, \mathcal{B}, d)
			\nonumber \\
			&\stackrel{{\bf IH }}{=}&
			DOw_i(N, {\bf 0}, \mathcal{B}, d) - DOw_j(N, {\bf 0}, \mathcal{B}, d)
			\nonumber \\
			&\stackrel{{\bf FwC }}{=}&
			DOw_i(N, v^d, \mathcal{B}, d) - DOw_j(N, v^d, \mathcal{B}, d).
		\end{eqnarray*}
		Hence, $f_i(N, v^d, \mathcal{B}, d) - DOw_i(N, v^d, \mathcal{B}, d)= f_j(N, v^d, \mathcal{B}, d)- DOw_j(N, v^d, \mathcal{B}, d)$ for any $i,j \in B_k \cap \mathcal{T}(v^d)$.
		Together with (\ref{eq1}) and (\ref{eq2}), this implies that
		\begin{equation*}\label{eq4}
			f_i(N, v^d, \mathcal{B}, d)=DOw_i(N, v^d, \mathcal{B}, d), \text{ for any }  i \in B_k \cap \mathcal{T}(v^d).
		\end{equation*}
 Thus, we have $f_i(N, v^d, \mathcal{B}, d)=DOw_i(N, v^d, \mathcal{B}, d)$ for all $i \in B_k$.
  Consequently, we deduce that $f(N, v, \mathcal{B}, d)=DOw(N, v, \mathcal{B}, d)$.
	\end{proof}
	
	\subsection{ A modification of Proposition 2}
	In this section, we present an alternative characterization by the axiom of Null Player Out for Preserving-Diversity Games  ({\bf NPOPD}) in Theorem \ref{thm1}  with the  axiom of Null Player for Diverse Games  ({\bf ND}).
	
	\begin{theorem}\label{theorem3}
		The Diversity Owen value is the unique value on $GD$ that satisfies Efficiency (E), Fairness within Component (\textbf{FwC}), Fairness through Diversity (\textbf{FD}), Independence from Non-Diverse Coalitions (\textbf{INDC}), and Null Player for Diverse Games (\textbf{ND}).
	\end{theorem}
	
	\begin{proof}
		{\em Existence}: According to \cite{2025_beal2} and Lemma \ref{lemma1}, one can easily check that the Diversity Owen value satisfies the five axioms of Theorem \ref{theorem3}.
		
		{\em Uniqueness}: Let \( f \) be a value on \( GD \) that satisfies the five axioms. We  show that \( f(N, v, \mathcal{B}, d) = DOw(N, v, \mathcal{B}, d) \)  for any $(N, v, \mathcal{B}, d)$. By {\bf INDC},  \( f(N, v, \mathcal{B}, d) = f(N, v^d, \mathcal{B}, d)\) and \( DOw(N, v, \mathcal{B}, d)$ $= DOw(N, v^d, \mathcal{B}, d)\). Thus it suffices to show that  $f(N, v^d, \mathcal{B}, d) = DOw(N, v^d, \mathcal{B}, d)$.
		Set
		\( \mathcal{I}(v^d) := \{ T \in D(N, \mathcal{B}, d) : \Delta_{v^d}(T) \neq 0 \} \). We show that $f(N, v^d, \mathcal{B}, d) = DOw(N, v^d, \mathcal{B}, d)$ by induction on \( |\mathcal{I}(v^d)|  \).
		
		{\em Induction basis}: If $|\mathcal{I}(v^d) |=0$, then every $i \in N$ is null in $v^d$. Since $v^d$ is a diverse game, by {\bf ND}, we have $f_i(N,v^d,\mathcal{B},d)=0=DOw_i(N,v^d,\mathcal{B},d)$ for all $i \in N$.
		
		{\em Induction hypothesis }({\bf IH}): Assume that $f(N, v^d, \mathcal{B}, d) = DOw(N, v^d, \mathcal{B}, d)$ for any  $(N, v, \mathcal{B}, d)$ when $|\mathcal{I}(v^d)| < m$.
		
		{\em Induction step }: Let us show that $f(N, v^d, \mathcal{B}, d) = DOw(N, v^d, \mathcal{B}, d)$ for any  $(N, v, \mathcal{B}, d)$ with $|\mathcal{I}(v^d)| = m$. By \textbf{FD} and \textbf{E},  we have
		\begin{equation}\label{eq5}
			\sum_{i \in B_p}f_i(N,v^d,\mathcal{B},d)=\sum_{i \in B_p}DOw_i(N,v^d,\mathcal{B},d),~ \text{for any } B_p \in \mathcal{B}.
		\end{equation}
		
		It remains to show that \( f_i(N, v^d, \mathcal{B}, d) = DOw_i(N, v^d, \mathcal{B}, d) \) for any $B_p  \in \mathcal{B}$ and any \( i \in B_p \).
		
		Let
$\gamma(B_p,v^d)=\{i \in B_p \mid i \in T, \forall T \in \mathcal{I}(v^d)\}$, and
$P(\mathcal{I}(v^d)) = \{ B_p \cap T \mid T \in \mathcal{I}(v^d),\, B_p \cap T \neq \emptyset \}$.
Since  every $T \in \mathcal{I}(v^d)$ such that $ T \in D(N, \mathcal{B}, d)$, $|P(\mathcal{I}(v^d))|\ge 1$. We distinguish the following three cases regarding $P(\mathcal{I}(v^d))$.
		
		{\bf Case 1}. $|P(\mathcal{I}(v^d))| = 1$. Then, $P(\mathcal{I}(v^d))=\{\gamma(B_p,v^d)\}$. Note that all $i,j \in \gamma(B_p,v^d)$ are
		symmetric in $v^d$. By {\bf FwC} and {\bf ND}, we have
		\begin{equation*}
			f_i(N, v^d, \mathcal{B}, d) - f_j(N, v^d, \mathcal{B}, d)\stackrel{{\bf FwC}}{=}f_i(N, {\bf 0}, \mathcal{B}, d) - f_j(N, {\bf 0}, \mathcal{B}, d)\stackrel{{\bf ND}}{=}0.
		\end{equation*}
Hence, $f_i(N, v^d, \mathcal{B}, d)=f_j(N, v^d, \mathcal{B}, d)$ for all $i,j\in \gamma(B_p,v^d)$.
Similarly, we have  \[DOw_i(N, v^d, \mathcal{B}, d)=DOw_j(N, v^d, \mathcal{B}, d) \ \ \mbox{for all $i,j\in \gamma(B_p,v^d)$}.\]
		Note that each $i \in B_p \backslash \gamma(B_p,v^d)$ is null in $v^d$. By {\bf ND}, $f_i(N, v^d, \mathcal{B}, d) =0= DOw_i(N, v^d, \mathcal{B}, d)$ for any $i \in B_p \backslash \gamma(B_p,v^d)$.
		Combining with (\ref{eq5}), we have \( f_i(N, v^d, \mathcal{B}, d) = DOw_i(N, v^d, \mathcal{B}, d) \) for any \( i \in B_p \in \mathcal{B} \).
		
		{\bf Case 2}. $|P(\mathcal{I}(v^d))| \geq 2$, and
		there exists $S_1,S_2\in P(\mathcal{I}(v^d))$ such that $S_1\cap S_2\neq \emptyset$ or $B_p\setminus (S_1\cup S_2) \neq \emptyset$, we consider the following two subcases. Suppose that \( S_1 \subseteq T_1 \) and \( S_2 \subseteq T_2 \), where \( T_1, T_2 \in \mathcal{I}(v^d) \).
		
		\textbf{Case 2.1}. \( S_1 \cap S_2 \neq \emptyset \). Then \( T_1 \cap T_2 \neq \emptyset \), and $S_1\setminus S_2\neq \emptyset$ or $S_2\setminus S_1\neq \emptyset$. Without loss of generality, we assume $S_1\setminus S_2\neq \emptyset$. Define $w_1 = \Delta_{v^d}(T_1)u_{T_1}$. Note that all players \( i, j \in S_1 \) or \( i, j \in B_p \setminus S_1 \) are symmetric in \( w_1  \). Then, by \textbf{FwC}, we obtain
		\begin{align*}
			f_i(N, v^d, \mathcal{B}, d) - f_j(N, v^d, \mathcal{B}, d)
			&\stackrel{\textbf{FwC}}{=} f_i(N, v^d - w_1, \mathcal{B}, d) - f_j(N, v^d - w_1, \mathcal{B}, d) \notag \\
			&\stackrel{\textbf{IH}}{=} DOw_i(N, v^d - w_1, \mathcal{B}, d) - DOw_j(N, v^d - w_1, \mathcal{B}, d) \notag \\
			&\stackrel{\textbf{FwC}}{=} DOw_i(N, v^d, \mathcal{B}, d) - DOw_j(N, v^d, \mathcal{B}, d).
		\end{align*}
Hence, $f_i(N, v^d, \mathcal{B}, d) - DOw_i(N, v^d, \mathcal{B}, d)=f_j(N, v^d, \mathcal{B}, d) - DOw_j(N, v^d, \mathcal{B}, d)
		$ for any  \( i, j \in S_1 \) or \( i, j \in B_p \setminus S_1 \).
		 In particular, take \( i \in S_1 \cap S_2  \) and \( j \in S_1 \setminus S_2  \), we have
		$f_i(N, v^d, \mathcal{B}, d) - DOw_i(N, v^d, \mathcal{B}, d)=f_j(N, v^d, \mathcal{B}, d) - DOw_j(N, v^d, \mathcal{B}, d).
		$
So,
		$
			f_i(N, v^d, \mathcal{B}, d) - DOw_i(N, v^d, \mathcal{B}, d)=f_j(N, v^d, \mathcal{B}, d) - DOw_j(N, v^d, \mathcal{B}, d) \text{ for all } i, j \in B_p.
		$
This,		together with (\ref{eq5}), implies that \[ f_i(N, v^d, \mathcal{B}, d) = DOw_i(N, v^d, \mathcal{B}, d)\ \  \mbox{for all \( i \in B_p \in \mathcal{B} \)}.\]
		
	\textbf{Case 2.2}. $B_p\setminus (S_1\cup S_2) \neq \emptyset$. Then, \( T_1 \cup T_2 \neq N \).  In view of \textbf{Case 2.1}, we may assume that $S_1\cap S_2 = \emptyset$. Define $w_1 = \Delta_{v^d}(T_1)u_{T_1}$, $w_2 = \Delta_{v^d}(T_2)u_{T_2}$.
		
		Take \( i \in S_1  \), \( j \in S_2  \) and \( k \in B_p \setminus (S_1 \cup S_2) \). Players $i$ and $k$ are symmetric in \( w_1  \), and players $j$ and $k$ are symmetric in \( w_2  \). According to \textbf{IH} and by applying \textbf{FwC}, we obtain
		\begin{align*}
		&f_i(N, v^d, \mathcal{B}, d) - DOw_i(N, v^d, \mathcal{B}, d)=f_k(N, v^d, \mathcal{B}, d) - DOw_k(N, v^d, \mathcal{B}, d).\\
		&f_j(N, v^d, \mathcal{B}, d) - DOw_j(N, v^d, \mathcal{B}, d)=f_k(N, v^d, \mathcal{B}, d) - DOw_k(N, v^d, \mathcal{B}, d).
		\end{align*}
Hence, $f_i(N, v^d, \mathcal{B}, d) - DOw_i(N, v^d, \mathcal{B}, d)=f_j(N, v^d, \mathcal{B}, d) - DOw_j(N, v^d, \mathcal{B}, d) \text{ for any } i, j \in B_p$. Combining this with (\ref{eq5}), we conclude that
		$f_i(N, v^d, \mathcal{B}, d) = DOw_i(N, v^d, \mathcal{B}, d)$ for all  $i \in B_p \in \mathcal{B}.$
		
		{\bf Case 3}. $|P(\mathcal{I}(v^d))| = 2$ such that $ P(\mathcal{I}(v^d)) = \{S_1, S_2\},  S_1\cap S_2 =\emptyset$, and $S_1\cup S_2 =B_p$.
		All players in \( S_1 \) or in \( S_2 \) are symmetric in \( v^d \). According to \textbf{IH} and by applying \textbf{FwC}, for any $i,j \in S_1$ or $i,j \in S_2$, we obtain
		\begin{align}\label{eq11}
			f_i(N, v^d, \mathcal{B}, d) - DOw_i(N, v^d, \mathcal{B}, d)&= f_j(N, v^d, \mathcal{B}, d) - DOw_j(N, v^d, \mathcal{B}, d).
		\end{align}
		
		For any \( i \in S_1 \) and any \( j \in S_2 \), define
		\[
		(v')^d = \sum_{T \in \mathcal{I}(v^d):S_2 \subseteq T} \Delta_{v^d}(T) u_T
		- \sum_{T \in \mathcal{I}(v^d): S_1 \subseteq T} \Delta_{v^d}(T) u_{(T \setminus \{i\}) \cup \{j\}}.
		\]
		
		In $(v')^d$, $P(\mathcal{I}((v')^d))$ contains $\{S_2,S_1'\}$ where $S_1'=(S_1\setminus \{i\}) \cup \{j\}$. Note that $S_1'\cap S_2 \neq \emptyset$.
		Since \( i \) and \( j \) are symmetric in \( v^d - (v')^d \), by {\bf Case 2} and {\bf FwC}, we obtain
		\begin{align*}
			f_i(N, v^d, \mathcal{B}, d) - f_j(N, v^d, \mathcal{B}, d)
			&\stackrel{{\bf FwC}}{=} f_i(N, (v')^d, \mathcal{B}, d) - f_j(N, (v')^d, \mathcal{B}, d) \\
			&\stackrel{{\bf Case 2}}{=} DOw_i(N, (v')^d, \mathcal{B}, d) - DOw_j(N, (v')^d, \mathcal{B}, d) \\
			&\stackrel{{\bf FwC}}{=} DOw_i(N, v^d, \mathcal{B}, d) - DOw_j(N, v^d, \mathcal{B}, d).
		\end{align*}
Hence, $f_i(N, v^d, \mathcal{B}, d) - DOw_i(N, v^d, \mathcal{B}, d)= f_j(N, v^d, \mathcal{B}, d)- DOw_j(N, v^d, \mathcal{B}, d)$ for any \( i \in S_1 \) and any \( j \in S_2 \).		
		Combining (\ref{eq11}), we obtain $f_i(N, v^d, \mathcal{B}, d) - DOw_i(N, v^d, \mathcal{B}, d)=f_j(N, v^d, \mathcal{B}, d) - DOw_j(N, v^d, \mathcal{B}, d)$  for any \( i, j \in B_p \). This, together with (\ref{eq5}),  yields that
		$
		f_i(N, v^d, \mathcal{B}, d) = DOw_i(N, v^d, \mathcal{B}, d)
		$ for all $i \in B_p\in \mathcal{B}$.
		
		In either case, we conclude that
		$f_i(N, v^d, \mathcal{B}, d) = DOw_i(N, v^d, \mathcal{B}, d)$ for all $i\in B_p\in \mathcal{B}$. Consequently, $f(N, v, \mathcal{B}, d)=DOw(N, v, \mathcal{B}, d)$.
	\end{proof}
	
	\section{On  Theorem \ref{thm2}}	\label{section4}

	In the uniqueness proof of Theorem \ref{thm2} in \cite{2025_beal2}, we first identify potential issues in Case 1. The original proofs are reproduced below.
	
	\begin{quote}
		\textbf{Case 1.} $|B_k|\geq 2$ and $|B_k|>d_k.$ ``
		Let $(N, v, B, d)\in GD$ be a game and $B_k\in \mathcal{B}$ a community such that $2 \leq|B_k|= \overline{t} + 1$. By applying \textbf{IBCOPPD}, we have
		\begin{align*}
			&f_i(N, v^d, \mathcal{B}, d) - f_j(N, v^d, \mathcal{B}, d)\\
			&= f_i(N\setminus \{j\}, (v^d)_{|N\setminus \{j\}}, \mathcal{B}_{|N\setminus \{j\}}, d) - f_j(N\setminus \{i\}, (v^d)_{|N\setminus \{i\}}, \mathcal{B}_{|N\setminus \{i\}}, d) \\
			&\stackrel{{\textit{IH}}}{=} DOw_i(N\setminus \{j\}, (v^d)_{|N\setminus \{j\}}, \mathcal{B}_{|N\setminus \{j\}}, d) - DOw_j(N\setminus \{i\}, (v^d)_{|N\setminus \{i\}}, \mathcal{B}_{|N\setminus {i}}, d)  \\
			&=DOw_i(N, v^d, \mathcal{B}, d) - DOw_j(N, v^d, \mathcal{B}, d).
		\end{align*}
		That is, $\cdots$"
	\end{quote}
	The induction hypothesis (\textbf{IH}) in Case 1 assumes that all communities satisfy $|B_k|>d_k$. However, in the induction step, after removing player $j$, the subgame $(N\setminus \{j\}, (v^d)_{|N\setminus \{j\}}, \mathcal{B}_{|N\setminus \{j\}}, d)$ might contain a community $B_k$ such that $|B_k|=d_k$.  In such cases, the \textbf{IH} condition $|B_k|>d_k$ no longer holds. Additionally, the validity of Case 2 relies on the result of Case 1.
	
	\subsection{The corrected proof of the uniqueness in Theorem \ref{thm2}}\label{subsection2}
	Before presenting the corrected proof of the uniqueness in Theorem \ref{thm2} in \cite{2025_beal2}, we introduce the following lemma.
	
	\begin{lemma}\label{lemma2}
		For any $(N, v, \mathcal{B}, d) \in GD$ and any value $f$ on $GD$, if $f$ satisfies {\bf E} and {\bf ED}, then for any $B_k \in \mathcal{B}$, $
		\sum_{i \in B_k} f_i(N, v^d, \mathcal{B}, d) = \sum_{i \in B_k} DOw_i(N, v^d, \mathcal{B}, d).$
	\end{lemma}
	\begin{proof}
		For any $B_k \in \mathcal{B}$, according to {\bf E} and {\bf ED}, we have $\sum_{i \in B_k}f_i(N, v^d, \mathcal{B}, d)=\frac{v(N)}{|M|}=\sum_{i \in B_k}DOw_i(N, v^d, \mathcal{B}, d)$.
	\end{proof}
	
	We give a corrected proof of the uniqueness in Theorem \ref{thm2}.
	\begin{proof}[Proof of Theorem~\ref{thm2}]
	To prove uniqueness, let $f$ be a value on $GD$ satisfying all the axioms in Theorem \ref{thm2}.	We show that \( f(N, v, \mathcal{B}, d) = DOw(N, v, \mathcal{B}, d) \) for any $(N, v, \mathcal{B}, d) \in GD$. By {\bf INDC}, it suffices to show that \( f(N, v^d, \mathcal{B}, d) = DOw(N, v^d, \mathcal{B}, d) \).
		
		For any $B_k \in \mathcal{B}$, if $|B_k| = 1$, by Lemma \ref{lemma2}, $f_i(N, v^d, \mathcal{B}, d)=DOw_i(N, v^d, \mathcal{B}, d)$ for any $i \in B_k$. We next assume that $|B_k| \geq 2$, and we consider the following two cases.
		
		{\bf Case 1}. $|B_k| \geq 2$ and $|B_k| = d_k$. As in Lemma \ref{lemma1}, we
		construct the extended game $(N \cup \{l\},(v)_{+l},\mathcal{B}_{+l},d)$ by adding an outside player $l\notin N$ to $B_k$.
		Note that player $l$ is null and every player $i \in B_k$ is out in $(v^d)_{+l}$. Then,  by {\bf IBCOPPD}, for each $i \in B_k$,
	\begin{eqnarray}\label{eq6}
		f_i(N \cup \{l\}, (v^d)_{+l}, \mathcal{B}_{+l}, d)
		 &-& f_i(N, v^d, \mathcal{B}, d) \nonumber \\
		&=&
		 f_l(N \cup \{l\}, (v^d)_{+l}, \mathcal{B}_{+l}, d) \nonumber \\
 &&- f_l\big((N \cup \{l\} )\setminus \{i\},\, ((v^d)_{+l})_{|(N \cup \{l\} ) \setminus \{i\}},\, (\mathcal{B}_{+l})_{|(N \cup \{l\} ) \setminus \{i\}},\, d\big) \nonumber \\
		&=&
		 f_l(N \cup \{l\}, (v^d)_{+l}, \mathcal{B}_{+l}, d),
	\end{eqnarray}
	where the first equality holds by {\bf IBCOPPD}, and the second equality holds since the game $((v^d)_{+l})_{|(N \cup \{l\} ) \setminus \{i\}} = \mathbf{0}$ and $l$ is null in this game.
	Thus, we have
	\begin{align*}
		|B_k|f_l(N \cup \{l\},(v^d)_{+l},\mathcal{B}_{+l},d)
		&\stackrel{(\ref{eq6})}{=}\sum_{i \in B_k}\big(f_i(N \cup \{l\}, (v^d)_{+l}, \mathcal{B}_{+l}, d) - f_i(N, v^d, \mathcal{B}, d)\big) \nonumber \\
		&=
		\sum_{i \in B_k \cup l }f_i(N \cup \{l\},(v^d)_{+l},\mathcal{B}_{+l},d)
		-\sum_{i \in B_k}f_i(N, v^d, \mathcal{B}, d)\nonumber \\&-f_l(N \cup \{l\},(v^d)_{+l},\mathcal{B}_{+l},d) \\ \nonumber
		&=\sum_{i \in B_k \cup l }DOw_i(N \cup \{l\},(v^d)_{+l},\mathcal{B}_{+l},d)
		-\sum_{i \in B_k}DOw_i(N, v^d, \mathcal{B}, d)\nonumber \\& -f_l(N \cup \{l\},(v^d)_{+l},\mathcal{B}_{+l},d)=-f_l(N \cup \{l\},(v^d)_{+l},\mathcal{B}_{+l},d),
	\end{align*}
where the third equality follows Lemma \ref{lemma2}.
	    Hence,
	\begin{equation}\label{eq12}
		f_l(N \cup \{l\},(v^d)_{+l},\mathcal{B}_{+l},d)=0.
	\end{equation}
		
	   Since any two players $i,j \in B_k$ are out in $(v^d)_{+l}$, applying {\bf IBCOPPD} again, we have
	   \begin{align*}
		f_i(&N \cup \{l\}, (v^d)_{+l}, \mathcal{B}_{+l}, d) \nonumber\\
		&= f_i(N \cup \{l\}, (v^d)_{+l}, \mathcal{B}_{+l}, d)
		- f_i\bigl( (N \cup \{l\}) \setminus \{j\}, ((v^d)_{+l})_{|(N \cup \{l\}) \setminus \{j\}}, (\mathcal{B}_{+l})_{|(N \cup \{l\}) \setminus \{j\}}, d \bigr) \nonumber\\
		&= f_j(N \cup \{l\}, (v^d)_{+l}, \mathcal{B}_{+l}, d)
		- f_j\bigl((N \cup \{l\}) \setminus \{i\}, ((v^d)_{+l})_{|(N \cup \{l\}) \setminus \{i\}}, (\mathcal{B}_{+l})_{|(N \cup \{l\})\setminus \{i\}}, d \bigr) \nonumber\\
		&= f_j(N \cup \{l\}, (v^d)_{+l}, \mathcal{B}_{+l}, d),
	    \end{align*}
	    where the first and third equations hold because the restricted games $((v^d)_{+l})_{|(N \cup \{l\}) \setminus \{j\}}$ and $((v^d)_{+l})_{|(N \cup \{l\}) \setminus \{i\}}$ are the null game, and $i,j$ are null in the respective restricted games, and the second equation
holds by {\bf IBCOPPD}. Hence, $f_i(N \cup \{l\}, (v^d)_{+l}, \mathcal{B}_{+l}, d)=f_j(N \cup \{l\}, (v^d)_{+l}, \mathcal{B}_{+l}, d)$ for any $i,j\in B_k$.
	    Then, by Lemma \ref{lemma2},
	    \begin{equation}\label{eq13}
	    	f_i(N \cup \{l\}, (v^d)_{+l}, \mathcal{B}_{+l}, d) =DOw_i(N \cup \{l\}, (v^d)_{+l}, \mathcal{B}_{+l}, d) \text{ for any } i \in B_k.
	    \end{equation}
	   Similarly, for any \( i \in B_k \), by {\bf IBCOPPD} and (\ref{eq12}), we have
	   \begin{align*}
	   	& f_i(N \cup \{l\}, (v^d)_{+l}, \mathcal{B}_{+l}, d)
	   	- f_i(N, v^d, \mathcal{B}, d) \\
	   	&= f_l(N \cup \{l\}, (v^d)_{+l}, \mathcal{B}_{+l}, d)
	   	- f_l\big((N \cup \{l\}) \setminus \{i\},\, ((v^d)_{+l})_{|(N \cup \{l\})\setminus \{i\}},\, (\mathcal{B}_{+l})_{|(N \cup \{l\}) \setminus \{i\}},\, d\big) \\
	   	&\stackrel{(\ref{eq12})}{=} 0.
	   \end{align*}
		Hence, $f_i(N \cup \{l\}, (v^d)_{+l}, \mathcal{B}_{+l}, d)
		= f_i(N, v^d, \mathcal{B}, d)$ for any $i \in B_k$.
		
		Combining with (\ref{eq13}), we conclude that
		$f_i(N,v^d,\mathcal{B},d)=DOw_i(N,v^d,\mathcal{B},d)$ for any $i \in B_k$.

		{\bf Case 2}. Assume that $|B_k|\geq 2$ and $|B_k|>d_k$. We proceed by induction on $|B_k|$. Induction basis \textbf{(IB)} and induction hypothesis \textbf{(IH)} are stated as in the proof of Proposition 4 in \cite{2025_beal2}.
		
		\textbf{Induction step:} Let $(N, v, B, d)\in GD$ be a game and $B_k\in \mathcal{B}$ a community such	that $2 \leq|B_k|= \overline{t} + 1$. By applying \textbf{IBCOPPD}, for any $i,j\in B_k$,  we have
		\begin{align*}
			&f_i(N, v^d, \mathcal{B}, d) - f_j(N, v^d, \mathcal{B}, d)\\
			&= f_i(N\setminus \{j\}, (v^d)_{|N\setminus \{j\}}, \mathcal{B}_{|N\setminus \{j\}}, d) - f_j(N\setminus \{i\}, (v^d)_{|N\setminus \{i\}}, \mathcal{B}_{|N\setminus \{i\}}, d) \\
			&\stackrel{{\bf Case~1}\text{~or~} \textbf{IH}}{=} DOw_i(N\setminus \{j\}, (v^d)_{|N\setminus \{j\}}, \mathcal{B}_{|N\setminus \{j\}}, d) - DOw_j(N\setminus \{i\}, (v^d)_{|N\setminus \{i\}}, \mathcal{B}_{|N\setminus {i}}, d)  \\
			&=DOw_i(N, v^d, \mathcal{B}, d) - DOw_j(N, v^d, \mathcal{B}, d).
		\end{align*}
		That is, $f_i(N, v^d, \mathcal{B}, d) - DOw_i(N, v^d, \mathcal{B}, d)= c$ for all $i\in B_k$, where $c$ is a constant. Moreover, $\sum_{i\in B_k}\big(f_i(N, v^d, \mathcal{B}, d)-DOw_i(N, v^d, \mathcal{B}, d)\big)=|B_k|c$, which implies that $c=0$. So, $f_i(N, v^d, \mathcal{B}, d)=DOw_i(N, v^d, \mathcal{B},d)$ for all $i \in B_k$.
		Therefore, we conclude that $f(N, v^d, \mathcal{B}, d)=DOw(N, v^d, \mathcal{B},d)$.
	\end{proof}
	

Finally, we observe that
 the axiom of intra-coalitional balanced contributions with out players for preserving diversity in Theorem \ref{thm2} can be weaken as follows.

\textbf{Weak intra-coalitional balanced contributions with out players for preserving diversity (IBCOPPD$^-$).}  For any diverse game $(N, v, \mathcal{B}, d) \in GD$ and any out players $i,j\in B_p \in \mathcal{B}$ with $i\neq j$,
$f_i(N, v, \mathcal{B}, d) - f_i(N\setminus \{j\}, v_{|N\setminus \{j\}}, \mathcal{B}_{|N\setminus \{j\}}, d)=f_j(N, v, \mathcal{B}, d) - f_j(N\setminus \{i\}, v_{|N\setminus \{i\}}, \mathcal{B}_{|N\setminus \{i\}}, d).$
Moreover, for all $k\in N, f_k(N, \textbf{0}, \mathcal{B}, d)=0$.

Note that \textbf{IBCOPPD} implies \textbf{IBCOPPD$^-$}, and therefore the Diversity Owen value satisfies \textbf{IBCOPPD$^-$}. Combining this with the uniqueness proof in Subsection~\ref{subsection2}, we can immediately derive the following corollary.
\begin{corollary}
	The Diversity Owen value is the unique value on $GD$ that satisfies Efficiency (\textbf{E}), Weak intra-coalitional balanced contributions with out players for preserving diversity (\textbf{IBCOPPD$^-$}), Equality through diversity(\textbf{ED}), and Independence from Non-Diverse Coalitions (\textbf{INDC}).
\end{corollary}

\end{document}